\documentclass[12pt]{article}

\usepackage{indentfirst, anysize, hyperref, graphicx, amssymb,amsmath, amssymb, amsmath}
\usepackage{amsthm}
\theoremstyle{plain}
\newtheorem{theorem}{Theorem}[section]
\newtheorem{lemma}[theorem]{Lemma}
\newtheorem{corollary}[theorem]{Corollary}
\newtheorem{proposition}[theorem]{Proposition}

\newtheorem{conjecture}[theorem]{Conjecture}

\theoremstyle{definition}
\newtheorem{definition}[theorem]{Definition}
\newtheorem{example}[theorem]{Example}

\theoremstyle{remark}

\newcommand{\mathcalc}{\mathcal{C}}
\newcommand{\mathcalv}{\mathcal{V}}
\newcommand{\rank}{\textrm{rank}}
\newcommand{\parityc}{C_\oplus}
\newcommand{\ranktwo}{\textrm{rank}_2}

\newcommand{\bsoplus}{\textrm{bs}_\oplus}
\newcommand{\wbs}{\textrm{wbs}}
\newcommand{\bs}{\textrm{bs}}
\newcommand{\wbsoplus}{\wbs_\oplus}

\newcommand{\dc}{\textrm{DC}}

\newcommand{\la}{\langle}
\newcommand{\ra}{\rangle}

\newcommand{\defeq}{\stackrel{\textrm{def}}{=}}
\def\squareforqed{\hbox{\rlap{$\sqcap$}$\sqcup$}}
\def\qed{\ifmmode\squareforqed\else{\unskip\nobreak\hfil
\penalty50\hskip1em\null\nobreak\hfil\squareforqed
\parfillskip=0pt\finalhyphendemerits=0\endgraf}\fi}

\newenvironment{proofof}[1]{\begin{trivlist}%
\item[]{\flushleft\em Proof of #1. }}
{\qed\end{trivlist}}

\newcommand{\comment}[1]{}

\begin{document}
  \title{On the parity complexity measures of Boolean functions}
\author{
{\normalsize Zhiqiang
Zhang}\footnote{Email:zhang@itcs.tsinghua.edu.cn. Supported  in part by the National Natural Science Foundation of China Grant 60553001, and the National Basic Research Program of China Grant 2007CB807900, 2007CB807901.} \\
{\normalsize Institute for Theoretical Computer Science, Center for Advanced Study,}\\
{\normalsize Tsinghua University, Beijing, 100084, P.R. China}\\
{\normalsize and}\\
{\normalsize Yaoyun Shi}\footnote{Email:shiyy@eecs.umich.edu.
Supported in part by National Science Foundation of the United
States under the grants
0347078  and 0622033.}\\
{\normalsize Department of Electrical Engineering and Computer Science,}\\
{\normalsize University of Michigan, 2260 Hayward Street, Ann Arbor,
MI 48109-2121, USA} \\
} \date{}

  \maketitle

\abstract{\noindent The parity decision tree model extends the decision tree model by allowing
the computation of a parity function in one step.
We prove that the deterministic parity decision tree complexity of any Boolean
function is polynomially related to the non-deterministic complexity of the function or its complement.
We also show that they are polynomially related to an analogue of the block
sensitivity.  We further study parity decision trees in their relations with an intermediate variant of the decision trees,
as well as with communication complexity.
}
\section{Introduction and summary of results}
\label{sec:intro}
\noindent The decision tree model  is perhaps the simplest
model of computation. It is, however, capable of capturing
the inherent complexity of many natural computational problems.
Its relations with other models of computation
have also proved to be useful. In this section, we will first review some definitions and key results on decision trees,
before we present a summary of our results.

Let $f:\{0, 1\}^n\rightarrow \{0, 1\}$ be a Boolean function throughout this paper,
unless specified otherwise. Formally, a decision tree algorithm for computing  $f$ is
a full binary tree $T$, labeled as follows: (1) each non-leaf vertex is labeled with an index $i\in\{1, 2, ..., n\}$
to the input bits, (2) each leaf  and each edge is labeled with either $0$ or $1$.
The computation of $T$ on an input $x\in\{0, 1\}^n$ is the path that starts at the root
and follows the $x_i$ edge from a vertex labeled with $i$. The leaf label reached by this path is the output
of $T$ on $x$. The depth of the tree is the worst-case complexity of the algorithm.
The minimum depth of all decision trees computing $f$ is the {\em deterministic decision tree complexity} of $f$,
denoted by $D(f)$.

A set of decision trees {\em non-deterministically computes} $f$, if for any input $x$,
$f(x)=1$, if and only if a decision tree from the set outputs $1$.
The {\em non-deterministic decision tree complexity} of $f$, denoted by $C^1(f)$,
is the smallest integer $k$ such that $f$ is computed non-deterministically by a set of depth-$k$
decision trees. Alternatively, $C^1(f)$ is characterized by the smallest integer $k$, such
that for any input $x$ with $f(x)=1$, there is a subset $S\subseteq\{1, ..., n\}$ such that any input $x'$
with the same value as $x$ on bits indexed by $S$ must also have $f(x')=1$. Thus $C^1(f)$ is also commonly
called the {\em $1$-certificate complexity}. The $0$-certificate complexity, $C^0(f)\defeq C^1(1-f)$,
and the {\em certificate complexity}, $C(f)\defeq \max\{C^0(f), C^1(f)\}$.

It follows straightforwardly from the definitions that $C(f)\le D(f)$. A key result~\cite{Beals:2001:QLB} is,
for any $f$,
\begin{equation}
\label{eqn:dcequiv}
D(f)\le C^1(f)C^0(f).
\end{equation}
Thus for any Boolean function, its deterministic complexity is polynomially related with its non-deterministic complexity
or that of its complement.
This is in sharp contrast with the fact that for Turing machine computations the corresponding question
of P versus NP remains open. In fact, several other complexity measures
such as randomized and quantum decision tree complexities are also known to be polynomially related to
the deterministic decision tree complexity. A comprehensive survey on the subject
is \cite{Buhrman:decision:survey}  by Buhrman and de Wolf.

If in a decision tree, each non-leaf vertex is labeled with a
$c\in\{0, 1\}^n$ instead, and the computation path follows the edge
labeled with $\la x, c\ra \defeq \sum_i x_i c_i\mod 2$, we call this
extended decision tree a {\em parity decision tree} and the
corresponding complexity as the {\em parity decision tree
complexity}, denoted by $D_\oplus(f)$. This model was first defined
in \cite{Kushilevitz:learning}, which derived some simple properties
of the complexity. The {\em parity certificate complexities},
$C^0_\oplus(f)$, $C^1_\oplus(f)$, and $C_\oplus(f)$, can be defined
in analogy to the certificate complexities (see
Definition~\ref{def:cert}). They measure the non-deterministic
parity decision tree complexities of $f$ (or $1-f$). Our first main
result is in analogy to~(\ref{eqn:dcequiv}).

\begin{theorem}\label{thm:dc}
For any Boolean function $f$, $D_\oplus(f)\le C_\oplus^0(f)C_\oplus^1(f)$.
\end{theorem}

The {\em block-sensitivity} of $f$, $\bs(f)$, is the smallest integer $k$ such that for any input $x\in\{0, 1\}^n$
there are $k$ pair-wise disjoint subsets of $\{1, ..., n\}$ such that flipping all bits in any of those subsets flips $f(x)$.
Nisan~\cite{Nisan:CREW} showed that, for any $f$,
\begin{equation}\label{eqn:bsequiv}
C(f)\le \bs^2(f).
\end{equation}
Together with the simple relation that $\bs(f)\le C(f)$, this result shows that $\bs(f)$
is polynomially related with $C(f)$, thus with $D(f)$. We define (in Definition~\ref{def:bs}) the {\em parity block sensitivity}
$\bs_\oplus(f)$, and show that a similar relation holds.

\begin{theorem}\label{thm:bs} For any Boolean function $f$, $\bs_\oplus(f)\le C_\oplus(f)\le \bs_\oplus^2(f)$.
\end{theorem}

The above three classes of parity complexities we study satisfy the following symmetry properties.
Let $c\in \{0,1\}^n$. The function obtained by shifting $f$ by $c$ is
$f_c: x\mapsto f(x+ c)$. Let $A$ be a linear transformation on
$\{0,1\}^n$ (as the $n$-dimensional linear space over the field
$\mathbb{F}_2$), $f_A$ is the function defined as $f_A(x) = f(Ax)$.
For any coset $H$ of  $\{0,1\}^n$ (i.e. a shift of a subspace), denote by $f|_H$ the restriction of $f$ on $H$. A
complexity measure $\Theta$ defined on Boolean functions is said to be
invariant under shift if $\Theta(f_c)=\Theta(f)$ for any $c\in\{0,1\}^n$. It is said to be invariant under rotation if $\Theta(f_A)=T(f)$ for any
invertible transformation $A$ over $\mathbb{F}_2^n$.

When $\Theta$ is invariant under shift and rotation, we can extend the domain of $\Theta$ to include any function $g$ defined on a coset $H$ of $\{0,1\}^n$.
For such a $g$, and a coset $H=c+S$ where $c\in\{0,1\}^n$ and $S$ is a subspace with basis $\{e_1,\cdots,e_m\}$, we define
$g':\{0, 1\}^m\rightarrow\{0, 1\}$ as follows,
\begin{equation}\label{eqn:change_g}
g'(x_1x_2\cdots x_m) \defeq g(c+x_1e_1+\cdots+x_me_m)\quad\textrm{for all $x\in\{0, 1\}^m$,}
\end{equation}
and extend $\Theta$ to $g$ by setting,
\begin{equation}\label{eqn:extend}
\Theta(g)\defeq \Theta(g').
\end{equation}
Then $\Theta(g)$ is well defined, as it is independent of the choice
of the basis and $c$ for $H$ due to $\Theta$ being invariant under
shift and rotation. We say a complexity measure $\Theta$ invariant
under shift and rotation is
 {\em monotone} if for any $n\ge1$, $f:\{0, 1\}^n\rightarrow\{0, 1\}$, and coset $H\subseteq\{0, 1\}^n$, $\Theta(f|_H)\leq \Theta(f)$.

All the classical complexity measures of Boolean functions such as
decision tree complexity, certificate complexity, and block sensitivity
are invariant only under shift but not under rotation.
The parity version complexities we study are, however, invariant under both shift and rotation,
and are monotone.

To contrast those two sets of complexity measures, we may ``symmetrize'' every classical complexity measure $\Theta$
to $\Theta_I$ by defining $\Theta_I(f)\defeq\min_{B} \Theta(f_B)$, where $B$ takes value from all invertible linear transformations.
A natural question is if each parity complexity is identical, or at least polynomially related,
to the rotation invariant version of the corresponding classical complexity.
We show that this is not the case. In this sense, the parity decision tree model is an inherently more powerful model
than the decision tree model.

\begin{theorem} \label{thm:exp}
  For infinitely many $n$, there exists $f_n:\{0, 1\}^n\rightarrow\{0, 1\}$, such that $D_\oplus(f_n) = O(\log n)$ and   $D_I(f_n)=\Theta(n)$.
\end{theorem}

Parity decision trees are closely related to the communication complexity of
XOR functions~\cite{Zhang:XOR}. Communication complexity is  a major branch of complexity theory
that studies the inherent communication cost for distributive computation.
The {\em deterministic communication complexity} of
$F:\{0, 1\}^n\times\{0, 1\}^n\rightarrow\{0, 1\}$, denoted by $\dc(F)$, is the smallest integer $k$, such that
there is a communication protocol between two parties Alice and Bob satisfying the following conditions:
(1) Alice's input is an $x\in\{0, 1\}^n$, and Bob's input is a $y\in\{0, 1\}^n$. (2) Alice and Bob take turn to send each other a message,
each message is determined by each party's input as well as the messages s/he has received previously.
(3) At the end of the protocol one party knows $F(x, y)$. (4) The total number of bits in the messages is $\le k$.
This model as well as its several variants have been extensively studied. For surveys, see \cite{Kushilevitz:1997:book, Sherstov:survey, Lee:survey}.

Determining $\dc(F)$ may be a highly nontrivial problem, even for the following class of functions of a simple structure.
A function $F:\{0, 1\}^n\times \{0, 1\}^n\rightarrow \{0, 1\}$ is called an {XOR function}~\cite{Zhang:XOR} if
for some $f:\{0, 1\}^n\rightarrow \{0, 1\}$, $F(x,y)=f(x +  y)$, for all $x, y\in\{0, 1\}^n$.
The computation of a parity decision tree $T$ for $f$ can be simulated by Alice and Bob for computing $F$:
each query $c$ is simulated by Alice and Bob computing $\la c, x\ra$ and $\la c, y\ra$, respectively, and exchange the outcomes.
\begin{proposition}
For any XOR function $F:\{0, 1\}^n\times\{0,1\}^n\to\{0, 1\}$ with $F(x, y)=f(x+y)$, $\dc(F)\le 2D_\oplus(f)$.
\end{proposition}

In Section~\ref{sec:comm}, we show that $C^1(f)$, times $\log n$, also gives an upper bound on the {\em non-deterministic communication
complexity} of $F$. A natural question is if those upper bounds are far from being tight. While we are not able to answer this question,
we conjecture they are. We also put forward a conjecture
that, if true, would  also imply the well-known Log-Rank Conjecture~\cite{Lovasz:logrank}
when restricted to XOR functions.

\section{Parity certificate complexity}\label{sec:cert}\noindent
We consider $\{0, 1\}^n$ as a $n$-dimensional vector space over
$\mathbb{F}_2$, the two-element finite field, as well as an Abelian group with respect
to the bit-wise XOR. Then a coset of $\{0, 1\}^n$ is a set $b + V$, where $b\in\{0, 1\}^n$ and $V$ is a subspace of $\{0, 1\}^n$.
The co-dimension of $b+V$ is $n-\dim(V)$. Equivalently, a coset is the set of solutions to a system of linear equations,
and the minimum number of the equations defining the same coset is the co-dimension.
Informally, the parity certificate complexity measures how many linear constraints have to be given on the input in order to fix the value of $f$.
\begin{definition}\label{def:cert}
Let $f:D\rightarrow\{0, 1\}$ be defined on $D\subseteq\{0, 1\}^n$,
and $x\in D$. A coset $S$ of $\{0, 1\}^n$ is called a {\em parity
certificate} of $f$ on $x$ if $s\in S$ and $f$ is constant on $S\cap
D$. The {\em size} of the certificate is defined to be the
co-dimension of $S$.   The minimum size of a parity certificate for
$x$ is denoted by $C_\oplus(f, x)$. The {\em parity certificate
complexity of $f$}, denoted by $C_\oplus(f)$, is $\max_{x}
C_\oplus(f, x)$.

A parity certificate $S$ is called a $0$- (or $1$-) parity certificate if $f(x)=0$ (or $f(x)=1$, respectively) for all $x\in S\cap D$.
The  {\em $0$- and $1$-parity certificate complexities of $f$} are
$C_\oplus^0(f)\defeq \max_{x:f(x)=0} C_\oplus(f, x)$, and $C_\oplus^1(f)\defeq \max_{x:f(x)=1}C_\oplus(f, x)$, respectively.
\end{definition}
If $f\equiv 0$ (or $f\equiv1$), then $C_\oplus^1(f)$ (or $C_\oplus^0(f)$, respectively) is not defined.
We may represent a parity certificate $S$ of size $T$ (or a coset $S$ of co-dimension $T$) by a pair $(C, r)$, where $C\in\{0, 1\}^{T\times n}$ and $r\in\{0, 1\}^T$,
such that $S=\{x: Cx=r\}$. It follows from the definitions that when $B\in\{0, 1\}^{n\times n}$ takes value
from invertible matrices,
\begin{equation}\label{eqn:CoplusC}
C_\oplus(f,x) = \min_B C(f_B, B^{-1}x).
\end{equation}
Similar relations between the $0$- and $1$-parity certificates/certificates also hold.
Note that $0$- and $1$-parity certificate complexity measure the non-deterministic parity decision tree complexity of
$f$ and $1-f$, respectively, with the non-deterministic parity decision tree complexity defined in analogy to the non-deterministic decision tree complexity.
Since any parity decision tree gives a certificate of size no more than the depth of the tree for any input, we have the following relation.
\begin{proposition} For any Boolean function $f$,
  $C_\oplus(f)\leq D_\oplus(f)$.
\end{proposition}

We now prove Theorem~\ref{thm:dc}, which states that $D_\oplus(f)\le C^0_\oplus(f)C^1_\oplus(f)$, for any $f$.
\begin{proofof}{Theorem~\ref{thm:dc}}
The idea of the proof is similar to that in~\cite{Beals:2001:QLB} for proving Inequality~(\ref{eqn:dcequiv}).
We give an algorithm that computes $f$ using no more than
$C^1_\oplus(f)C^0_\oplus(f)$ queries.

Fix an input $x_0$. For a sequence of cosets $(C_1, r_1)$, $(C_2, r_2)$, ..., define
$V_i\defeq \{x: C_jx=C_jx_0, j=1, 2, ..., i\}$ for $i\ge1$ and $V_0\defeq\{0, 1\}^n$.
By definition, $V_0\supseteq V_1\supseteq V_2 \supseteq\cdots$.
The algorithm will examine a sequence of $1$-parity certificates, $(C_1, r_1)$, $(C_2, r_2)$,
..., that it constructs incrementally from an initially empty sequence.
It proceeds as follows:
For $i=1, 2, ...$, if $f|_{V_{i-1}}$ is constant, output that constant and terminate.
Otherwise, extend the current sequence of $1$-parity certificates with a new one $(C_i, r_i)$ for $f|_{V_{i-1}}$
of the smallest size. Since $f|_{V_{i-1}}$ is not constant, such a $1$-parity certificate exists.
Query the rows in $C_i$. If the answers agree with $r_i$, return $1$. Otherwise continue with $i$ incremented
by $1$.

The algorithm clearly outputs the correct answer. Since restricting a function on a subset does not increase $C_\oplus^1$,
at most $C_\oplus^1(f)$ queries are made in the $i$th iteration, for each $i$.
We prove that $f|_{V_T}$ is constant for some $T\le C^0_\oplus(f)$.
Assume otherwise and fix an $x'_0\in V_T$ with $T=C^0_\oplus(f)$ and $f(x'_0)=0$. We argue that for each $i$, $1\le i\le T$,
\begin{equation}\label{eqn:decrease}
C_\oplus(f|_{V_{i}}, x'_0)\le C_\oplus(f|_{V_{i-1}}, x'_0)-1.
\end{equation}
Fix a parity certificate $(C, r)$ for $f|_{V_{i-1}}$ containing $x'_0$ and of the smallest size. Since the linear system
$\{C_i x= r_i, Cx=r\}$ does not have a solution  in $V_{i-1}$ but the system $\{C_i x= r_i\}$ does
(by the definition of $(C_i, r_i)$ being a $1$-parity certificate for $f|_{V_{i-1}}$, which is non-constant),
the row space of $C$ has a non-empty intersection
with the space spanned by the rows of $C_1, ..., C_i$. Assume without loss of generality that the intersection
is spanned by the first $k$ rows, for some $k\ge1$, in $C$ (otherwise, apply an appropriate invertible matrix on both sides of $Cx=r$),
and denote the sub-matrix of $C$ and $r$ containing those rows by $C'$ and $r'$, and the remaining portions by $C''$ and $r''$.
Any $x\in V_i$ satisfying $C''x=r''$  must have $C'x=C'x_0=C'x'_0=r'$, thus $Cx=r$, implying $f(x)=0$.
Thus $(C'', r'')$ is a parity certificate containing $x'_0$ for $f|_{V_i}$, and
Eqn.~(\ref{eqn:decrease}) holds. Consequently, $C_\oplus(f, x'_0) \ge T+ C_\oplus(f|_{V_T}, x'_0)\ge T+1
> C_\oplus^0(f)$, a contradiction.
Therefore $f|_{V_T}$ is constant for some $T\le C^1_\oplus(f)$, and the algorithm uses no more than $C_\oplus^1(f)C_\oplus^0(f)$
number of queries.
\end{proofof}

\section{Parity block sensitivity}
\label{sec:bs}
\noindent
Recall that the {\em block sensitivity of $f$ on an input $x$}, $\bs(f, x)$, is the smallest
integer $k$, such that there exist $S_1, S_2, ..., S_k\subseteq\{1, 2, ..., n\}$ that are pair-wise disjoint,
and for each $i$, $1\le i\le k$, $f(x) \ne f(x^{S_i})$, where $x^{S_i}\in\{0, 1\}^n$ is obtained from $x$ by flipping
each bit indexed by $S_i$. The {\em block sensitivity of $f$}, $\bs(f)$, is $\max_x \bs(f, x)$.
We define the parity analogues of those concepts. First define weak parity block sensitivity $\wbs(f,x)$ similar to the definition of parity certificate complexity.

\begin{definition}
The {\em weak parity block sensitivity of $f$ on $x$} is
\[\wbsoplus(f,x)\defeq\min_B\  \bs(f_B,B^{-1}x).\]
The {\em weak parity block sensitivity of $f$} is
\[\wbsoplus(f)\defeq\max_x\ \wbsoplus (f,x).\]
\end{definition}

Note that $\wbsoplus(f)$ is invariant under shift and rotation, so we can extend it to  functions defined on a coset
through Eqn.~(\ref{eqn:extend}).
%Then we can talk that whether it is monotone. ``Monotonicity'' is a natural requirement for a ``Complexity''.
The following example shows that $\wbsoplus(f)$ is not monotone.

\begin{example}
 Consider $f(x_1,x_2,x_3)=x_1\oplus(x_2\vee x_3)$. For any input $x$, we can always
  choose a basis $\{e_1,e_2,e_3\}$ such that $f(x+e_i)=f(x)$, $i=1,2,3$. For example, when $x=011$ we can choose the
  basis  $\{010,001,111\}$. For such bases, any sensitive block contains at least two base vectors.
  So there is at most one sensitive block, implying $\wbsoplus(f,x)\leq 1$. But with $H=\{x:x_1=0\}$,
  $f|_H(x_2,x_3)= x_2\vee x_3$.
  This is the OR function on two variables, of which the parity block sensitivity is 2 at 0.
  Thus for this $f$, $\wbsoplus(f) < \wbsoplus(f|_H)$.
\end{example}

We modify $\wbsoplus$ to a parity complexity measure
by taking maximum over all restrictions to cosets. Then it will be invariant under shift and rotation, and is monotone.

\begin{definition}\label{def:bs} For a Boolean function $f:\{0, 1\}^n\rightarrow\{0, 1\}$,
its {\em parity block sensitivity}, $bs_\oplus(f)$, is
\[\bs_\oplus(f)=\max_H\  \wbsoplus(f|_H),\]
where $H$ takes value from the cosets of $\{0, 1\}^n$.
\end{definition}
Similar to Inequality~(\ref{eqn:bsequiv}), Theorem~\ref{thm:bs} implies that
the parity block sensitivity is polynomially related to parity certificate complexity.
We give below the proof for the Theorem, which states that $\bs_\oplus(f)\le C_\oplus(f)\le \bs_\oplus^2(f)$ for any $f$.
The proof idea is also similar to that for proving~(\ref{eqn:bsequiv}) in~\cite{Nisan:CREW}.
\begin{proofof}{Theorem~\ref{thm:bs}}
Since $C_\oplus$ is monotone,  to prove $\bs_\oplus(f)\leq C_\oplus(f)$, it suffices to prove $\wbsoplus(f,x)\leq
C_\oplus(f)$, for any $x$. This follows straightforwardly from the definition, the relation between block sensitivity
and certificate complexity, and Eqn.~(\ref{eqn:CoplusC}):
\[\wbsoplus(f,x)=\min_B \bs(f_B, B^{-1}x) \leq \min_B C(f_B,
B^{-1}x)= C_\oplus(f,x).\]

We prove the second inequality by showing  $C_\oplus(f)\leq
\wbsoplus(f)\bs_\oplus(f)$. Since the three quantities are both invariant under shift, we assume
without loss of generality that $C_\oplus(f)$ is achieved at $x=0$. Also assume without loss of generality
that $f(0)=0$. Since $C_\oplus(f, x) = C_\oplus(f_B, B^{-1}x)$ for any invertible $B$ and any $x$,
we can further assume without loss of generality that $b\defeq\wbsoplus(f, 0)=\bs(f, 0)$.
Let $S_1$, $S_2$, ..., $S_b\subseteq\{1, 2, ..., n\}$
be a collection of disjoint and minimal sets achieving $\bs(f, 0)$.
Consider $S=\{x: x_i=0, i\in S_1\cup S_2\cup \cdots S_b \}$.
Then $S$ is a parity certificate for $f$, as otherwise there would be a block $S'\subseteq
\left(\{1, ..., S\}-{\bigcup_{i=1}^bS_i}\right)$
such that $f(0^{S'})=1$, contradicting that $b=\bs(f, 0)$.

Fix an $i$, $1\le i\le b$. Let $m=|S_i|$ and  $S_i=\{a_1,a_2, ..., a_m\}$.
Consider $f|_{H_i}$, where $H_i\defeq\{x: x_j=0, j\in \{1, 2, ..., n\}- S_i\}$.
Then $f|_{H_i}:\{0, 1\}^m\rightarrow\{0, 1\}$ and
\[ f|_{H_i}(y) = f\left(\sum_{i=1}^m y_i e_{a_i}\right),\quad\textrm{for all $y\in\{0, 1\}^m$}.\]
Since $S_i$ is minimal, for any $S'_i\subseteq S_i$, $f(0^{S'_i})=1$ if and only if $S_i'=S_i$.
Thus $f|_{H_i}(y)$ is the AND function on $m$ variables.
Therefore $\wbsoplus(f|_{H_i})=m$. Consequently, $m\le \bsoplus(f)$.
Thus $C_\oplus(f)=C_\oplus(f, 0)\le \sum_{i=1}^b |S_i|\le \wbsoplus(f,0)\bsoplus(f)$, implying
$C_\oplus(f)\le \wbsoplus(f)\bsoplus(f)$.
\end{proofof}

\section{The gap between parity measures and symmetrized classical measures}\noindent
In this section, we prove Theorem~\ref{thm:exp},
which states that for infinitely many $n$, there exists $f_n:\{0, 1\}^n\rightarrow\{0, 1\}$, such that $D_\oplus(f_n) = O(\log n)$ and   $D_I(f_n)=\Theta(n)$.
 We will define the desired function $f_n$
by a random parity decision tree of logarithmic depth, then show that there exists such a parity decision tree
of which the function requires linear certificate complexity, thus linear decision tree complexity.

For $A\in\{0, 1\}^{m\times n}$, $s\in\{0, 1\}^n$, define
\[\tau_A(s)\defeq\min\{|s+v| :\textrm{$v\in$ row space of $A$}\}.\]
We will need the following lemma to lower bound the certificate complexity.
\begin{lemma}\label{lm:exp}
Let $f:\{0, 1\}^n\rightarrow\{0, 1\}$, $s\in\{0, 1\}^n$ and
$f(x)=\langle x, s\rangle$ for all $x$ in a coset $H=(A, r)$. Then
$C(f)\geq \tau_A(s)$. In particular, $D(f)\ge \tau_A(s)$.
 \end{lemma}
\begin{proof}
Choose an arbitrary $x_0\in H$. Let $\ell\defeq C(f, x_0)\le C(f)$.
Suppose that $E\in\{0, 1\}^{\ell \times n}$ describes a certificate.
That is, each row in $E$ contains all $0$ but a single $1$, and all
$x'$ with $Ex'=Ex_0$ must have $f(x')=f(x_0)$.

Now consider two sets of equations on the unknown $y\in\{0, 1\}^n$:
\begin{equation*}
\left\{\begin{array}{rcl}
Ey &=& Ex_0\\
Ay &=& r\\
\la s, y\ra &=&\la s, x_0\ra
\end{array}\right.
\quad\textrm{and}\quad
\left\{\begin{array}{rcl}
Ey &=& Ex_0\\
Ay &=& r\\
\la s, y\ra &=& 1-\la s, x_0\ra
\end{array}\right. \ .
\end{equation*}

The first set of equations has a solution (e.g. $y=x_0$) but not the
second set, since all $y$ satisfying $Ay=r$ must have $\la s,
y\rangle=\langle s, x_0\rangle$. This is possible only when $s$ is
in the span of the rows in $E$ and in $A$. Thus for some $v$ in the
row space of $A$, $s+v$ is in the row space of $E$. Thus
$\tau_A(s)\le \ell$. Therefore, $\tau_A(s)\le C(f)$. That $D(f)\geq
\tau_A(s)$ follows from the fact that $C(f)\le D(f)$.
\end{proof}

We are ready to prove Theorem~\ref{thm:exp}.
\begin{proofof}{Theorem~\ref{thm:exp}}
Let $n=2^k$. We construct a function $f$ with $n$ variables decided
by a parity decision tree $T$ of depth $k+4$.
For $1 \leq i\leq k+3$, all the $i$-th layer
 nodes are labeled by $e_i\defeq0^{i-1}10^{n-i}$.
 The $t$-th node of the last layer before the output, $1\le t\le 8n$, is labeled by  a random
$s_t\in\{0,1\}^{n}$. The answer to this query $\langle x, s_t\rangle$ is the output.

Fix an invertible matrix $B$. Then  $f_B$ is computed by the parity tree that replaces each query $c$ in $T$ by $B^Tc$.
In this parity decision tree, the inputs that arrive at a node with query $s'_t\defeq B^Ts_t$ form a coset  $H_t=(C_t, r_t)$ of co-dimension $k+3$,
and $f_B(x)= \la x, s'_t\rangle$ for all $x\in H_t$.
By Lemma~\ref{lm:exp},
$D(f_B)\geq \tau_{C_t}(s'_t)$.

For each $v$ in the row space of $C_t$,  $s'_t+v$ is uniformly distributed. Thus by Hoeffding's Inequality,
$\Pr(|s'_t+v|\leq n/4)\leq e^{-n/8}$.
Thus
$$\Pr(\tau_{C_t}(s'_t)\leq n/4)\leq 2^{k+3}e^{-n/8}=8ne^{-n/8}.$$
There are $8n$ independently chosen $s_j$, thus
$$\Pr(D(f_B)\geq n/4)\geq 1-\left(8ne^{-n/8}\right)^{8n}=1-(8n)^{8n}e^{-n^2}.$$
There are at most $(2^n)^n=2^{n^2}$ different transformations $B$(the exact number  is $\Pi_{i=0}^{n-1}(2^{n-i}-1)$). Therefore,
$$P(\min_BD(f_B)\geq n/4)\geq 1-(8n)^{8n}e^{-n^2}\cdot2^{n^2}=1-(8n)^{8n}\left(\frac{2}{e}\right)^{n^2} \rightarrow 1.$$
This implies that when $n$ is large enough, almost all the functions $f$ computed by the above
parity trees have $D_I(f)=\min_BD(f_B)\geq n/4$. In contrast, the parity
decision tree complexity of these $f$ is no more than $k+4=\log_2n+4$.
\end{proofof}

The following corollary follows from the polynomial relations among
certificate complexity and block sensitivity with decision tree
complexity and their analogy for parity
complexities.\begin{corollary} For infinitely many $n$, there exists
a $n$-variate $f_n$ such that the gaps between $C_\oplus(f)$ and
$C_I(f)$ and between $\bs_\oplus(f)$ and $\bs_I(f)$ are exponential.
\end{corollary}

\section{Connection with communication complexities}\noindent
\label{sec:comm}
In a {\em non-deterministic communication protocol} for computing $F:\{0, 1\}^n\times \{0, 1\}^n\rightarrow\{0, 1\}$,
Alice or Bob may non-deterministically choose from a set of strategies for the rest of the communication.
We say that the protocol computes $F$ if for any $(x, y)$, $F(x, y)=1$ if and only if for some choice in the non-deterministic steps
the protocol outputs $1$. Denote  the {\em non-deterministic communication complexity} of $F$
by $N^1(F)$. A fundamental result by Aho, Ullman and Yannakakis~\cite{AUY} is $\dc(F) = O(N^1(F)N^1(1-F))$, a
relation similar to those about decision tree complexity and parity decision tree complexity.
The main result of this section relates $N^1(F)$ with $C^1_\oplus(f)$ for XOR functions $F$ with $F(x, y)=f(x+y)$.

\begin{theorem} \label{nc}For any XOR function $F(x,y)=f(x\oplus y)$,
 $N^1(F)\leq \parityc^1(f)\log n$.\end{theorem}

To prove this result, we will make use of the following notion.
\begin{definition}
A set $\mathcalc$ of $1$-parity certificates for $f$ is called {\em essential} if (1)
for any $x$ with $f(x)=1$ there is an element in $\mathcalc$ containing $x$,  (2) no element is a subset of the union
of all the other elements, and (3) any element is of a size $C^1_\oplus(f)$.
\end{definition}

Clearly there exists an essential set of $1$-parity certificates, as one could start with
one smallest $1$-parity certificate for each $x$, increase its size to $C^1_\oplus(f)$ if necessary,
and remove any element contained in the union of the rest of the set.

\begin{proofof}{Theorem~\ref{nc}} Let $d=\parityc^1(f)$. Fix an essential set $\mathcalc=\{(C_i, r_i): 1\le i\le K\}$ of $1$-parity certificates.
The following is a simple non-deterministic communication protocol for $F$.
bits of communication: Alice non-deterministically chooses $(C_i, r_i)\in\mathcalc$,
sends $i$, as well as $C_ix$. Bob checks if $C_ix + C_iy = r_i$. He accepts if yes, rejects otherwise.
The correctness of the protocol follows from the definition of $1$-parity certificate and the assumption that
$\mathcalc$ contains a $1$-parity certificate for any $1$-input. The total cost is $d +\lceil \log_2( K+1)\rceil$.
Lemma~\ref{lm:cover} below shows that  $K=n^{O(d)}$.
Thus $N^1(F)=O(d\log n)$.
\end{proofof}

\begin{lemma}\label{lm:cover}
Let $\mathcalc$ be an essential set of $1$-parity certificates
for $f:\{0, 1\}^n\rightarrow\{0, 1\}$ and $d=C^1_\oplus(f)$. Then
$|\mathcalc|\leq n^{O(d)}$.
\end{lemma}

\begin{proof}
Let $P$ be the number of pairs $(x, C)$ that $x\in C$ and $ C\in\mathcalc$.
Since $|C|=2^{n-d}$ for each $C$,
\begin{equation}\label{eqn:P}
P=2^{n-d}\ |\mathcalc|.
\end{equation}
For each $x\in\{0, 1\}^n$, let $S_1$, $S_2$, ..., $S_k\in\mathcalc$ be those
that contains $x$. Then $V_i\defeq x+S_i$, $1\le i\le k$, are $n-d$-dimensional subspaces none of which
is a subset of the union of the rest. We show below any such set of subspaces must
have $k=n^{O(d)}$. Thus $P =  2^n n^{O(d)}$. Together with Eqn.~(\ref{eqn:P}), this implies the conclusion
that $|\mathcalc|=n^{O(d)}$.

Let $C_i\in\{0, 1\}^{d\times n}$ such that $V_i=\{x: C_ix=0\}$, $1\le i\le k$.
For any $i$, let $x_i\in V_i$ be such that $x_i\not\in  \bigcup_{j\ne i} V_j$.
Then $C_ix_i=0$, but $C_jx_i\neq 0$ for all $j\neq i$.
Consider a $kd\times k$ matrix
$$G=\left[\begin{array}{c}C_1\\ C_2 \\ \cdots \\ C_k\end{array}\right][x_1,x_2\cdots,x_k].$$

Let $\ranktwo$ denote the rank over filed $\mathbb{F}_2$.
Then
$\ranktwo(G)\leq n$ from the above factorization of $G$.
Represent $G$ by a $k\times k$ block matrix
$a_{ij}$, where each block $a_{ij}$ is a $d\times 1$ vector.

For each $t$, $1\le t\le d$, define the $k\times k$ submatrix $G^t = [a_{ij}^t]_{1\le i, j\le k}$,
where $a_{ij}^t$ is the $t$-th element of $a_{ij}$. Since $G^t$ is a submatrix of $G$,
$\ranktwo(G^t)\leq\ranktwo(G)\leq n$.

Let $M=G^1\vee G^2 \vee \cdots \vee G^d$ be the entry-wise
conjunction of $G^1, G^2, \cdots, G^d$.
Notice that for any matrix $A$ and $B$,
$A\vee B = A+ B+ A\odot B$, where $A\odot B$ is the entry-wise product of $A$ and $B$.
Since $\ranktwo(A\odot B)\le \ranktwo(A)\ranktwo(B)$, we have
\[\ranktwo(A\vee B)\le \ranktwo(A)+\ranktwo(B)+\ranktwo(A\odot B) \leq 3\ranktwo(A)\ranktwo(B).\]
Thus $\ranktwo(M) < (3n)^d$. On the other hand,
from the fact that $a_{ij}=0$ iff $i=j$, $M=I-J$, where $I$ is the identity matrix
and $J$ the all $1$ matrix. Thus $\ranktwo(M)\geq
\ranktwo(I)-\ranktwo(J)=k-1$. This implies
$k=|\mathcalv|\leq (3n)^d$.
\end{proof}

The following conjecture, if true, would imply that $\dc(F)$ is polynomially related to $D_\oplus(f)$ (as well as $C_\oplus(f)$),
by the Aho-Ullman-Yannakakis Theorem and Theorem~\ref{thm:dc}.
\begin{conjecture}\label{conj:parity_optimal}
For any XOR function $F$ based on $f$, $N^1(F)=\Omega(C^1_\oplus(f))$.
\end{conjecture}

A major open problem on deterministic communication complexity is the Log-Rank Conjecture~\cite{Lovasz:logrank}.
Denote by $\rank(F)=\rank([F(x, y)]_{x, y\in\{0, 1\}^n})$, where $\rank(\cdot)$ is the rank over the reals.
The Log-Rank Conjecture states that
\begin{equation}\label{logrank}
\dc(F)=\log^{O(1)} \rank(F),\quad\textrm{for any $F$}.
\end{equation}
The study of XOR functions is partly motivated by the Log-Rank
Conjecture. Denote by \[\|\hat f\|_0 =  | \{ \hat f_w \ne 0:
w\in\{0, 1\}^n\}|,\] where \[\hat f_w=\frac{1}{2^n} \sum_{x\in\{0,
1\}^n} (-1)^{\la x, w\ra} f(x)\] is the Fourier coefficient of $f$
on $w$. Then for any XOR function $F$ based on $f$,  $\rank(F)
=\|\hat f\|_0$. Our conjecture below, if true, would imply the
Log-Rank Conjecture on XOR functions.

\begin{conjecture}\label{conj:fd}
For any Boolean function $f:\{0, 1\}^n\rightarrow\{0, 1\}$,
$D_\oplus(f)$ and $C_\oplus(f)$ are
polynomially related with $\log\|\hat f\|_0$.
\end{conjecture}

\section{Acknowledgments}\noindent
We thank Xiaoming Sun and Andrew Yao for helpful discussions.

\bibliographystyle{abbrv}
\bibliography{my}
\end{document}